\documentclass[11pt]{amsart}
\usepackage[english]{babel}
\usepackage[ansinew]{inputenc}
\usepackage{amsmath, amsthm, amsfonts, amssymb}
\usepackage{latexsym}
\usepackage{bbm}
\usepackage{mathrsfs}
\usepackage{tikz}
\usepackage{graphicx}
\usepackage{multirow}
\usepackage{subfig}
\usepackage{youngtab}

\newtheorem{theorem}{Theorem}[section]
\newtheorem{proposition}[theorem]{Proposition}
\newtheorem{lemma}[theorem]{Lemma}

\theoremstyle{definition}
\newtheorem{definition}[theorem]{Definition}

\newtheorem{example}[theorem]{Example}
\newtheorem{problem}[theorem]{Problem}

\theoremstyle{remark}
\newtheorem{remark}[theorem]{Remark}

\numberwithin{equation}{section}
\numberwithin{theorem}{section}

\def\N{\mathbb{N}}
\def\Z{\mathbb{Z}}

\def\R{\mathbb{R}}
\def\C{\mathbb{C}}

\def\det{\mathrm{det}}
\def\per{\mathrm{per}}
\def\sgn{\mathrm{sgn}}
\def\GL{\mathrm{GL}}
\def\SL{\mathrm{SL}}
\def\VNP{\mathrm{VNP}}
\def\VP{\mathrm{VP}}
\def\dc{\mathrm{dc}}
\def\wt{\mathrm{weight}}
\def\rank{\mathrm{rank}}
\newcommand{\chara}{{\mathrm{char}}}
\def\Sym{\mathrm{Sym}}
\newcommand{\Det}{\Omega} 
\newcommand{\ol}[1]{\overline{#1}}
\def\e{\epsilon}
\def\la{\lambda}
\def\pleth{\mathrm{pleth}}
\def\mdet{\mathrm{multdet}}
\def\mult{\mathrm{mult}}
\def\ot{\otimes} 
\def\sP{\#\mathrm{P}}
\def\stab{\mathrm{stab}}
\def\id{\mathrm{id}}

\def\diag{\mathrm{diag}}
\def\NP{\mathrm{NP}}
\def\coNP{\mathrm{coNP}}
 
\title{Permanent versus determinant, obstructions,\\ and Kronecker
  coefficients$^*$}

\thanks{$^*$
This is an elaboration of a series of three lectures at the 
75th S\'eminaire Lotharingien de Combinatoire and 
XX Incontro Italiano di Combinatorica Algebrica in
Bertinoro, Italy, September 6--9, 2015}

\author{Peter B\"urgisser$\dagger$}

\thanks{$\dagger$
Institute of Mathematics, Technische Universit\"at Berlin, 
pbuerg@math.tu-berlin.de.
Partially supported by DFG grant BU 1371/3-2.}

\date{}



\begin{document}
\maketitle

\begin{abstract}
We give an introduction to some of the 
recent ideas that go under the name ``geometric complexity theory''.
We first sketch the proof of the known upper 
and lower bounds for the determinantal complexity of the permanent.
We then introduce the concept of a representation theoretic obstruction, 
which has close links to algebraic combinatorics, and we explain some of the insights
gained so far. In particular, we address very recent insights on the
complexity of testing the positivity of Kronecker coefficients. 
We also briefly discuss the related asymptotic version of this question.
\end{abstract}

\section{Motivation} 

The {\em determinant polynomial} is defined as 
$$
 \det_n := \det(X) := \sum_{\pi\in S_n} \sgn(\pi) \prod_{i=1}^n x_{i\pi(i)} ,
$$
where $x_{ij}$ are variables over a field $K$. 
The determinant derives its importance from the fact that it defines a 
group homomorphism $\det\colon \GL_n(K) \to K^\times$ due to 
$$
 \det(X\cdot Y) = \det (X)\, \det(Y) .
$$
It is highly relevant for computational mathematics that the determinant 
has an efficient computation. For instance, by using Gaussian
elimination, it can be computed with $O(n^3)$ arithmetic operations.

The definition of the {\em permanent polynomial} looks similarly as
for that of the determinant: 
$$
 \per_n := \per(X) := \sum_{\pi\in S_n} \prod_{i=1}^n x_{i\pi(i)} ,
$$
but without the sign changes. The permanent has less 
symmetries: 
$\per(X\cdot Y) = \per(X) \per(Y)$ holds 
if $X$ is a product of a permutation 
and a diagonal matrix, or if $Y$ is so; but in general, the 
multiplicativity property is violated. 
Also, for the permanent, there is no known efficient computation. 
We don't know whether there is a polynomial time algorithm for computing it. 
The permanent often shows up in algebraic combinatorics and statistical physics 
as a generating function in enumeration problems. 

In computer science, the permanent is known as a universal (or complete) problem 
in a class of weighted enumeration problems. One says that the family 
$(\per_n)$ of permanents is $\VNP$-complete. This theory was created in 1979
by L.~Valiant~\cite{Val:79b}. 
See \cite{buer:00-3,mapo:08} for more information.   

Proving that computing $\per_n$ requires superpolynomially many
arithmetic operations in~$n$ 
is considered the holy grail of algebraic complexity theory. This essentially amounts to proving 
the separation $\VP\ne\VNP$ of complexity classes. 
This separation is an ``easier'' variant of the famous P$\,\ne\,$NP problem. 

\section{Determinantal complexity}

Note that 
$\per\begin{bmatrix} a & b \\ c & d \end{bmatrix} = \det \begin{bmatrix} a & -b \\ c & d \end{bmatrix}$.  
P\'olya~\cite{polya:13} asked in 1913 whether such a formula is also
possible for $n\ge 3$, i.e., whether there is a sign matrix $[\epsilon_{ij}]$ 
such that $\per_n = \det[\epsilon_{ij} x_{ij}]$. 
This was disproved by Szeg{\H o}~\cite{szego:13} in the same year. 
Marcus and Minc~\cite{marcus-minc:61} strengthened this result by showing that there is no 
matrix $[f_{pq }]$ of linear forms $f_{pq}$ in the variables $x_{ij}$ such that 
$\per_n = \det[f_{pq}]$. 

\medskip

But what happens if we allow for the determinant a larger matrix? 
\medskip

We can express $\per_3$ as the determinant of a matrix of size 7, 
whose entries are constants or variables, cf.~\cite{grenet:11}: 
$$
 \per_3 = \det\begin{bmatrix} 
0&0&0&0&x_{33}&x_{32}&x_{31}\\
x_{11}&1&0&0&0&0&0\\
x_{12}&0&1&0&0&0&0\\
x_{13}&0&0&1&0&0&0\\
0&x_{22}&x_{21}&0&1&0&0\\
0&x_{23}&0&x_{21}&0&1&0\\
0&0&x_{23}&x_{22}&0&0&1
\end{bmatrix}.
$$
\medskip

\begin{definition}
The {\em determinantal complexity} $\dc(f)$ of a polynomial~$f\in K[x_1,\ldots,x_N]$ 
is the smallest~$s$ such that there exists a square matrix $A$ of size $s$, whose entries are affine 
linear functions of $x_1,\ldots,x_N$, such that $f= \det(A)$.  
Moreover, we write $\dc(m) := \dc(\per_m)$. 
\end{definition}

We clearly have $\dc(2)=2$. 
By the above formula, $\dc(3)\le 7$. 
Recent work showed the optimality: $\dc(3)=7$; cf.\ 
\cite{hu-ik:14,al-bo-ve:15}.

\subsection{An upper bound}

The following nice upper bound is due to Grenet~\cite{grenet:11}, 
based on ideas in Valiant~\cite{Val:79b}. 

\begin{theorem}[Grenet]\label{th:grenet}
We have $\dc(m) \le 2^m -1$.  
\end{theorem}

\begin{proof}
1. We first give the determinant of a matrix $A$ of size $m$ 
a combinatorial interpretation. 
We consider the complete directed graph with the node set 
$[m]:=\{1,2,\ldots m\}$ and the edges $(i,j)$ carrying the weight~$a_{ij}$. 
Moreover, we interpret a permutation $\pi$ of $[m]$ as the collection of 
their disjoint cycles (including loops for the fixed points) and call this 
a {\em cycle cover} $c$ of the digraph. We write $\sgn(c):=\sgn(\pi)$. 
The weight of $c$ is defined as 
the product of the weights of the edges occurring in $c$. 

Then we see that $\det(A)$ equals the sum of the signed weights 
over all cycle covers of the digraph: 
$$
 \det(A) = \sum_c \sgn(c) \wt(c) .
$$

2. We build now a digraph $P_m$ (see Figure~\ref{fig:jan}). 
Its node set is the power set $2^{[m]}$ of $[m]$.
For each $S\in 2^{[m]}$ of size $i-1$, where $1\le i \le m$,  
and $j\in [m]\setminus S$,  
we form a directed edge from $S$ to $S\cup\{j\}$ of weight $x_{ij}$. 
It is easy to see that 
$$
 \per_m(X) = \sum_\pi \wt(\pi) ,
$$
where the sum is over all directed paths $\pi$ going from
$\emptyset$ to $[m]$.
(We define the weight of $\pi$ as the product of the weights 
of its edges.)   

\begin{figure}[h]
\begin{center}
\includegraphics[width=\textwidth]{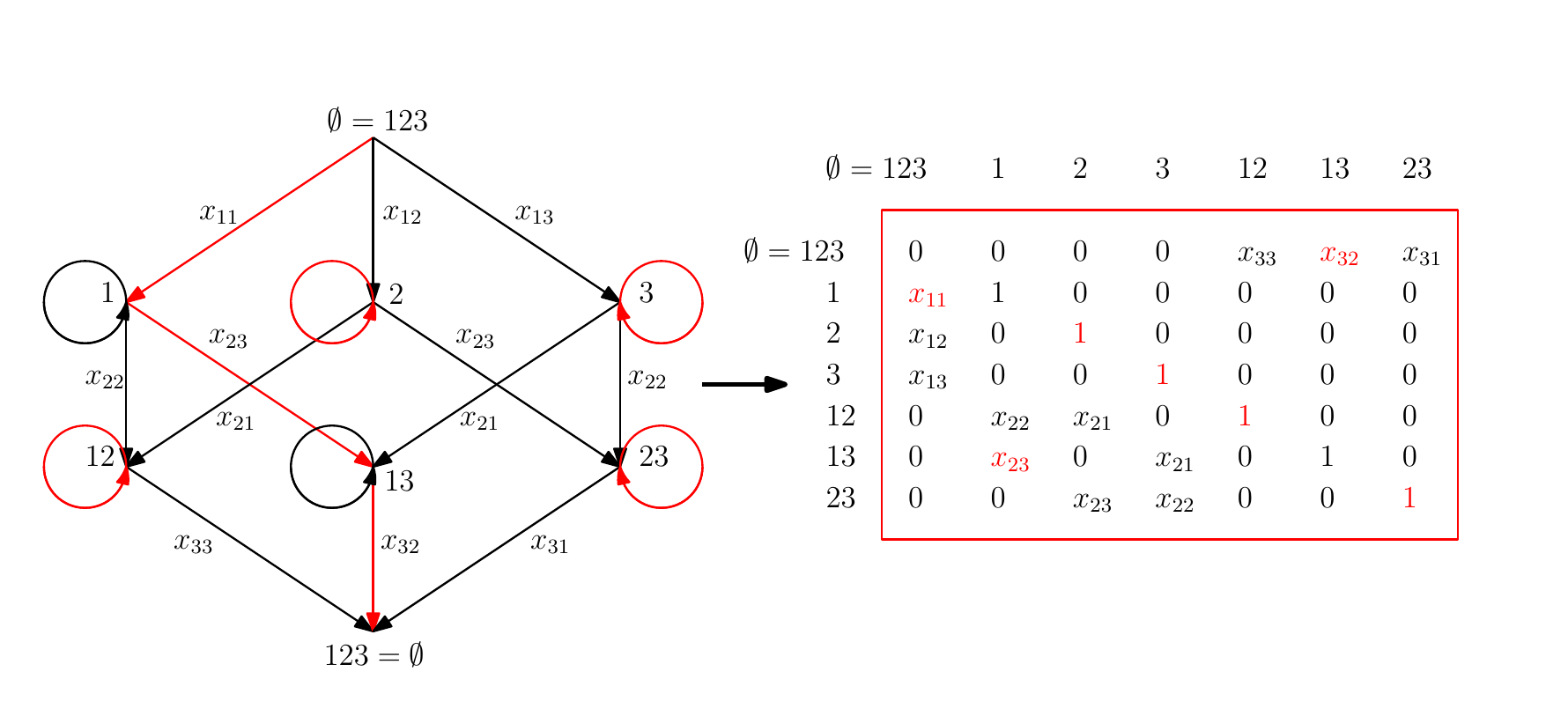}
\end{center}
\caption{{\small The construction for $m=3$. Courtesy of J.~Draisma~\cite{draisma-siam:15}.}}\label{fig:jan}
\end{figure}

We perform some modifications in this graph: we add loops 
of weight one at all nodes $S\in 2^{[m]}$ different from $\emptyset$ and $[m]$, 
and we identify the node $\emptyset$ with the node $[m]$. 
Let $A$ denote the weighted adjacency matrix of the resulting digraph. 
Its size is $2^m -1$. 

Then it is easy to see that we obtain a weight preserving bijection between the 
set of directed paths $\pi$ between $\emptyset$ and $[m]$ in the 
original digraph and the set of cycle covers $c_\pi$ in the modified digraph.
We obtain 
$$
 (-1)^{m-1} \per_m(X) = \sum_\pi (-1)^{m-1} \wt(\pi) = \sum_c \sgn(c) \wt(c) ,
$$
which shows that indeed $\dc(m) \le 2^m -1$. 
\end{proof}

Landsberg and Ressayre~\cite{lare:15} recently proved that 
the representation $\per_m = \det (A)$ in the proof of
Theorem~\ref{th:grenet} is optimal among all representations 
respecting ``half of the symmetries'' of $\per_m$.

\subsection{A lower bound}

The following result due to Mignon and Ressayre~\cite{mignon-ressayre:04} is the 
best known lower bound for $\dc(m)$, 
except for a recent improvement over $K=\R$ due 
to Yabe~\cite{yabe:15}, which states 
$(m-1)^2 +1 \le \dc(m)$. 

\begin{theorem}[Mignon and Ressayre]\label{th:mr}
We have $m^2/2 \le \dc(m)$ if $\chara K= 0$. 
\end{theorem}

\begin{proof}
The idea is to consider the {\em Hessian} $H_f$ of a polynomial 
$f\in K[x_1,\ldots,x_N]$:
$$
 H_f := \left[ \frac{\partial^2f}{\partial x_\alpha \partial x_\beta} \right]_{1\le\alpha,\beta\le N} .
$$
We note that 
$\frac{\partial^2\det_n}{\partial x_{ij} \partial x_{k\ell}}$ 
equals the minor of $X$ obtained by deleting the rows $i,j$ and 
colums $j,\ell$. 

The following is straightforward to verify using the chain rule.

\begin{lemma}\label{le:HT}
If we perform an affine linear transformation on 
$f\in K[x_1,\ldots,x_N]$, namely,
$$
 F(x_1,\ldots,x_M) := f\Big( L\cdot\begin{bmatrix}x_1\\ \vdots\\ x_M\end{bmatrix} + b\Big),
 \quad L\in K^{N\times M}, b\in K^N ,
$$
then 
$$
  H_F(x) = L^T H_f(Lx + b) L .
$$
\end{lemma}

Now we assume $\dc(m) \le n$. 
This means we have a representation
\begin{equation}\label{mr-uno}
 \per_m(X) = \det(A(X)) ,
\end{equation}
where $A(X)$ is of size $n$ and the entries of $A$ are affine linear in the $X$-variables. 
Lemma~\ref{le:HT} implies 
\begin{equation}\label{mr-due}
 H_\per(X) = L^T H_\det(A(X)) L,
\end{equation}
where $L \in K^{n^2\times m^2}$ is the matrix of the linear map
corresponding to the affine map $A$. 

We substitute in \eqref{mr-uno} the matrix $X$ 
by some $M\in K^{m\times m}$ which satisfies $\per(M)=0$, 
and we set $N:= A(M)$. 
Then, 
$$
  0 = \per (M) = \det(A(M)) = \det(N),
$$
so that $N$ is rank deficient. 
Moreover, \eqref{mr-due} implies 
\begin{equation}\label{eq:conclusion}
 \rank H_\per(M) \ \le\ \rank H_\det (N) .
\end{equation}

The determinant is special in the sense that its 
Hessian has small rank at rank deficient matrices $N$. 

\begin{lemma}\label{le:H-rank-det}
The rank of $H_\det(N)$ at a matrix $N\in K^{n\times n}$ only depends on the 
rank $s$ of $N$. If $s<n$, then 
$$
 \rank H_\det(N) \ \le\ 2n .
$$
\end{lemma}

\begin{proof}(Sketch)
$\det\colon K^{n\times n} \to K$ is an invariant with respect to the action
of $\SL_n\times\SL_n$ on $K^{n\times n}$ via 
$(S,T) \cdot N := SNT^{-1}$. Using Lemma~\ref{le:HT} one sees that 
$H_\det\colon K^{n\times n} \to K$ is an invariant under this action as well.
This implies the first assertion. 

For the second assertion, take $N$ in normal form 
($s$ ones on the diagonal and zeros otherwise) and compute the rank $H_\det(N)$. 
\end{proof}

In contrast, the permanent has the following property.

\begin{lemma}\label{le:ex-M}
There exists $M\in K^{m\times m}$ such that $\per(M) = 0$ and 
$H_\per(M)$ has rank $m^2$. (Here we assume $\chara K= 0$.) 
\end{lemma}

\begin{proof}(Sketch)
One may take 
$$
 M=\begin{bmatrix}
1-m&1&\cdots&1\\
1&1&\cdots &1\\
\vdots&\vdots& & \vdots\\
1&1&\cdots &1
\end{bmatrix},
$$
which satisfies $\per(M)=0$. 
It is elementary, though a bit cumbersome, to verify
that $H_\per(M)$ has full rank.
\end{proof}

Using Lemma~\ref{le:H-rank-det} and Lemma~\ref{le:ex-M} in \eqref{eq:conclusion}, 
we obtain 
$$
 m^2 = \rank H_\per(M) \ \le\ \rank H_\det (N) \le 2n
$$
and the assertion follows.
\end{proof}

We remark that~\cite{cai:2010} has an extension of Theorem~\ref{th:mr} to positive characteristic. 

\section{An attempt via algebraic geometry and representation theory}
\label{se:attempt}

How could we possibly prove better lower bounds on $\dc(m)$? 

\subsection{The determinant variety $\Det_n$}

We assume $K=\C$ in the following. 
We consider $\Sym^n\C^{n^2}$ as the space of homogeneous polynomials of degree~$n$ 
in $n^2$ variables. The group $\GL_{n^2}$ acts on $\Sym^n\C^{n^2}$ by linear substitution. 

\begin{definition}
The {\em orbit} $\GL_{n^2}\cdot \det_n$ is obtained by applying all possible 
invertible linear transformations to $\det_n$.  
The {\em orbit closure} of $\det_n$, 
$$
\Det_n := \ol{\GL_{n^2}\cdot \det_n} \subseteq \Sym^n\C^{n^2}, 
$$ 
is its closure  with respect to the Euclidean topology. 
We call $\Det_n$ the {\em determinant variety}. 
\end{definition}

\begin{example}
1. If $n=2$, we have 
$$
\GL_{4}\cdot \det_2 = \{\mbox{quadratic forms of rank $4$}\}, \quad
   \Det_2 := \Sym^2\C^{4} .
$$ 

2. We have for $\e\to 0$
$$
 \det\begin{bmatrix} x_{11} & \e x_{12} \\ \e x_{21} & x_{22} \end{bmatrix} 
 = x_{11} x_{22} - \e^2 x_{12} x_{21} \longrightarrow x_{11} x_{22} \in \Det_2 
  \quad\mbox{ for $\e\to 0$.}
$$
The latter observation generalizes to any $n$ and hence $x_{11}\cdots x_{nn} \in \Det_n$. 
\end{example}

\begin{remark}
For $n=3$, the boundary of $\Det_n$ has been determined recently~\cite{lahu:16}, 
but for $n=4$ it is already unknown. 
\end{remark}

The following observation allows to study $\Det_n$ 
with the methods of algebraic geometry. 

\begin{theorem}\label{th:tops}
$\Det_n$ is Zariski-closed, i.e., the zero set of a system of polynomial equations.
\end{theorem}

This is a consequence of a general principle saying that for 
any constructible subset of $\C^N$, the Zariski closure and the closure 
with respect to the Euclidean topology coincide, see Mumford~\cite[\S2.C]{mumf:95}. 


We make now the following observation. 

Suppose $\dc(m) \le n$ with $m>2$, say 
$\per_m(X) = \det (A(X))$, where $A(X)$ is of size~$n$, with affine linear entries 
in $x_{11},\ldots,x_{mm}$. 
(By Theorem~\ref{th:mr} we have $m<n$.)  
Homogenizing this equation with the additional variable $t$, 
we obtain
\begin{equation}\label{eq:perder}
t^{n-m} \per_m(X) = t^n \per_m\Big(\frac{1}{t}X\Big) =  t^n \det\Big(A\Big(\frac{1}{t}X\Big)\Big) 
   =  \det\Big(t A\Big(\frac{1}{t}X\Big)\Big) . 
\end{equation}
The entries of the matrix $t A(\frac{1}{t}X)$ 
are linear forms in $t$ and the $X$-variables.
We call $t^{n-m} \per_m(X)$ the {\em padded permanent}. 

The $n^2$ entries of $t A(\frac{1}{t}X)$,  arranged as a vector, 
may be thought of as being obtained by multiplying 
some matrix $L \in \C^{n^2 \times (m^2 +1)}$ with $(x_{11},\ldots,x_{mm},t)^T$. 
Now think of $t$ as being one of the variables in 
$\{x_{11},\ldots,x_{nn}\} \setminus \{x_{11},\ldots,x_{mm}\}$.
Then 
$L\cdot (x_{11},\ldots,x_{mm},t)^T = L' \cdot (x_{11},\ldots,x_{nn})^T$, 
where $L'$ is obtained by appending $n^2-m^2-1$ zero columns to $L$. 
We thus see that $t^{n-m} \per_m(X)$ is obtained from $\det_n$ by the 
substitution $L'$. Since $\GL_{n^2}$ is dense in $\C^{n\times n}$, 
we can approximate $L'$ arbitrarily closely by invertible 
matrices and hence we obtain 
$$
 t^{n-m} \per_m(X)  \in \Det_n .
$$
Mulmuley and Sohoni~\cite{gct1} proposed to prove that 
$t^{n-m} \per_m(X)  \not\in \Det_n$, which is stronger 
than $\dc(m) > n$, but which has the benefit that this problem 
can be naturally approached by tools from algebraic geometric. 
In particular, methods from geometric invariant theory 
can be brought into play. 

The basic strategy for proving lower bounds is now to exhibit a polynomial 
function 
$$ 
 R\colon\Sym^n\C^{n^2} \to \C
$$
that vanishes on $\Det_n$, but not 
on the padded permanent $t^{n-m} \per_m(X)$.  
Theorem~\ref{th:tops} tells us that this strategy ``in principle'' must work, 
but how on earth could we find such a function $R$?

The idea is to exploit the symmetries. The determinant variety $\Det_n$ 
clearly is invariant under the action of the group $\GL_{n^2}$ on $\Sym^n\C^{n^2}$. 
We consider the vanishing ideal
$$
 I(\Det_n) =\{ R \mid \mbox{$R$ vanishes on $\Det_n$} \} ,
$$
which is invariant under the action of $\GL_{n^2}$.
We bring now the representation theory of $\GL_{n^2}$ 
into play and try to understand which types of irreducible 
$\GL_{n^2}$-modules appear in $I(\Det_n)$. 

\subsection{A primer on representation theory}

Our treatment here is extremely brief. Basically, we just recall 
definitions and introduce notations. E.g., see
\cite{FH:91} for more information on this classical topic.

It is well-known that the isomorphism types of irreducible (rational) 
$\GL_{n^2}$-modules can be labelled by highest weights, 
which we can view as $\la\in\Z^{n^2}$ such that 
$\la_1\ge\cdots \ge \la_{n^2}$.  
The Schur-Weyl module $V_\la=V_\la(\GL_{n^2})$ denotes an 
irreducible $\GL_{n^2}$-module of highest weight $\la$. 

If $\la_{n^2}\ge 0$, then $\la$ is a partition of {\em length} 
$\ell(\la):= \#\{ i \mid \la_i \ne 0 \} \le n^2$ and {\em size} 
$|\la| := \sum_i \la_i$. 
We briefly write $\la\vdash_{n^2} |\la|$ for this. 

\begin{example}
1. If $\la=(\delta,\ldots,\delta)$ for $\delta\in\Z$, then 
$V_\la=\C$ with the operation 
$g\cdot 1 = \det(g)^{\delta}$. 

2. If $\la=(\delta,0,\ldots,0)$ for $\delta\in\N$, then 
$V_\la = \Sym^{\delta} \C^{n^2}$. 
\end{example}

The group $\GL_{n^2}$ acts on $\Sym^d\Sym^n\C^{n^2}$, 
and we are interested in its isotypical decomposition:
\begin{equation}\label{eq:plethysm}
 \Sym^d\Sym^n\C^{n^2} = \bigoplus_{\la\vdash dn} \pleth_n(\la) V_\la .
\end{equation}
The arising multiplicities $\pleth_n(\la)\in\N$ are called {\em plethysm coefficients}. 

\begin{remark}
The decomposition of $\Sym^d\Sym^n\C^{2}$ describes the invariants and covariants of binary forms 
of degree $n$. This was a subject of intense study in the 19th century and 
famous names like Cayley, Sylvester, Clebsch, Gordan, Hermite, Hilbert, $\ldots$ 
are associated with it (e.g., see \cite{schur:68,sturmfels:93}). 
However, in the above situation of forms of many variables, 
little is known. 
\end{remark}

We now go back to the vanishing ideal of $\Det_n$ and ask for the isotypical 
decomposition of the degree $d$ component of its 
vanishing ideal $I(\Det_n)$: 
\begin{equation}\label{eq:mdet}
 I(\Det_n)_d = \bigoplus_{\la\vdash dn} \mdet_n(\la) V_\la .
\end{equation}
Our goal is to get some information about the arising 
multiplicities $\mdet_n(\la)$. 
It will be convenient to say that the elements of the isotypical component 
$\mdet_n(\la) V_\la$ contain the {\em equations for $\Det_n$ of type $\la$}.  
Representation theory tells us that the equations ``come in modules''.
The {\em multiplicity $\mdet_n(\la)$}, multiplied by $\dim V_\la$, 
tells us how many linearly independent 
equations of type $\la$ there are. 

In order to say something about $\mdet_n(\la)$, we recall the following 
crucial quantity.

\begin{definition}\label{def:KronCoeff}
Let $\la_i\vdash_{m_i} N$, $i=1,2,3$, be three partitions of $N$ 
with length $\ell(\la_i) \le m_i$. Their {\em Kronecker coefficient} 
is defined as the multiplicity of the irreducible 
$\GL_{m_1}\times\GL_{m_2}\times \GL_{m_3}$-module in 
$\Sym^N\big(\C^{m_1}\ot \C^{m_2}\ot \C^{m_3}\big)$: 
$$
 k(\la_1,\la_2,\la_3) := \mult\Big(
  V_{\la_1}\ot V_{\la_2}\ot V_{\la_3}, \Sym^N \big(\C^{m_1}\ot \C^{m_2}\ot \C^{m_3}\big) \Big) .
$$
\end{definition}

It is well-known that, by Schur-Weyl duality, there is also an interpretation of Kronecker coefficients 
in terms of representations of the symmetric group $S_N$: we have 
$$
 k(\la_1,\la_2,\la_3) = \dim \big( [\la]\ot [\mu]\ot [\nu] \big)^{S_N},
$$ 
where $[\la]$ denotes an irreducible $S_N$-module of type $\la$ (Specht module). 

Unfortunately, despite being fundamental, Kronecker coefficients are not well understood.
We believe that they should count some efficiently describable objects, but such a description 
has so far only be achieved in special cases (notably, if one of the partitions is a hook, cf.~\cite{blasiak:12}).
Computer science has developed models to express this question in a rigorous way. 
We encode partitions as lists of binary encoded integers.

\begin{problem}
Is the function $(\la_1,\la_2,\la_3) \mapsto k(\la_1,\la_2,\la_3)$ 
in the complexity class $\sP$? 
\end{problem}

We will see that 
the case where two of the three partitions are equal and of rectangular shape 
$n\times d = (d,\ldots,d)$ ($n$ times), is of special interest to us. 
We therefore define 
\begin{equation}\label{def:RKron}
 k_n(\la) := k(\la,n\times d, n\times d) \quad 
 \mbox{ for $\la\vdash dn$.} 
\end{equation}

\subsection{Obstructions} 

The coordinate ring of $\Det_n$ consists of the restrictions 
of polynomial functions to $\Det_n$ and can be described as 
$$
 \C[\Det_n] := \C\big[\Sym^n\C^{n^2}\big]/I(\Det_n) . 
$$
The multiplicity of the irreducible $\GL_{n^2}$-module $V_\la$ in $\C[\Det_n]$
can be expressed as 
\begin{equation}\label{def:ktilde}
 \tilde{k}_n(\la) := \pleth_n(\la) - \mdet_n(\la) ,
\end{equation}
which we shall call {\em GCT-coefficients}. 
The following theorem, which is due to Mulmuley \& Sohoni~\cite{gct2}, 
shows that $\tilde{k}_n(\la)$ is upper bounded by the special Kronecker coefficients $k_n(\la)$. 
A refinement of this result can be found in~\cite{BLMW:11}.  

\begin{theorem}[Mulmuley and Sohoni]\label{th:ms}
We have $\tilde{k}_n(\la) \le k_n(\la)$ for $\la\vdash_{n^2} dn$.  
\end{theorem}


We explain now how we intend to apply this theorem for the purpose of lower bounds.
(Currently, this plan could not yet be realized, and we will explain below some of the 
difficulties encountered with its realization.)

Suppose that $k_n(\la) = 0$. Then Theorem~\ref{th:ms} implies that 
$\mdet_n(\la) = \pleth_n(\la)$. Looking at 
the decompositions~\eqref{eq:plethysm} and~\eqref{eq:mdet}, 
we infer that any polynomial 
$R\in\Sym^d\Sym^n\C^{n^2}$ of type~$\la$ vanishes on the 
determinant variety $\Det_n$. 
If we are lucky, and additionally, 
some $R$ of type $\la$ satisfies 
$R(t^{n-m}\per_m) \ne 0$, then we can conclude that 
the padded permanent $t^{n-m}\per_m$ does not lie in $\Det_n$ .
Therefore the lower bound $\dc(m) >n$ would follow. 

We call such a partition $\la$ an 
{\em (occurrence) obstruction proving $\dc(m) > n$}.  
\medskip

The nonvanishing condition for $R$ has the following consequences.
First of all, we must have $\pleth_n(\la)>0$. 
Moreover, we have the following constraints 
on the shape of $\la$. 

\begin{theorem}[Landsberg and Kadish]\label{th:kl}
If there exists $R\in\Sym^d\Sym^n\C^{n^2}$ of type $\la\vdash_{n^2} dn$ 
such that $R(t^{n-m}g) \ne 0$ for 
some  form $g$ of degree~$m$ in $\ell\le n^2$ variables, then 
$\ell(\la) \le \ell +1$ and $\la_1 \ge |\la| ( 1- m/n)$.
\end{theorem}

The first assertion is from \cite{BLMW:11} 
and the second is from \cite{kadish-landsberg:14}. 
We omit the proof. 

Hence an obstruction $\la$ has relatively few rows and 
almost all of its boxes are in its first row.
More specifically, in our situation, we have $\ell=m^2$. 
Therefore, a hypothetical sequence $(\la^m)$ of obstructions certifying 
at least $m^2/2\le \dc(m)$ must satisfy 
$\ell(\la^m)\le m^2 +1$ and 
$\lim_{m\to\infty} \la^m_1/|\la^m| = 1$. 

To further simplify, let us now forget about the nonvanishing 
of $R$ on the padded permanent and 
make the following definition. 

\begin{definition}
An {\em obstruction for forms of degree~$n$} 
is a partition $\la\vdash_{n^2} dn$, for some $d$, such that 
$k_n(\la) = 0$ and $\pleth_n(\la) >0$.
\end{definition}
 
\begin{proposition}
Assume there exists an obstruction~$\la$ for forms of degree~$n$ 
with $\ell=\ell(\la)$ rows. Then a generic polynomial 
$f\in\Sym^n\C^{\ell}$ of degree~$n$ in $\ell$ variables
satisfies $\dc(f) > n$. 
\end{proposition}

\begin{proof} 
The assumption $\pleth_n(\la) >0$ implies that there exists some 
homogeneous polynomial function
$R\colon\Sym^n\C^{n^2} \to \C$ of type $\la$; cf.~\eqref{eq:plethysm}.
Moreover, we may assume that the restriction of~$f$ to $\Sym^n\C^{\ell}$ does not vanish. 
(For this, one needs to know that $\pleth_n(\la)$ does not change when removing 
zeros from $\la$.)
By Theorem~\ref{th:ms}, $k_n(\la)=0$ implies 
$\tilde{k}_n(\la)=0$ and hence $R$ vanishes on ~$\Det_n$;
cf. \eqref{eq:plethysm}. 
For a generic $f\in\Sym^n\C^{\ell}$ we have $R(f)\ne 0$. 
Hence $f \not\in \Det_n$, which proves that $\dc(f) > n$.
\end{proof}

\begin{example}(Ikenmeyer~\cite{ike:12b})\label{ex:obstruction}
$\la=(13,13,2,2,2,2,2)$ is an obstruction for forms of degree~$3$ in $7$ variables.
Indeed, $|\la|= 36 = 12\cdot 3$, $\ell(\la) =7$ and one can check 
with computer calculations that 
$\pleth_3(\la) = 1$ and $k_3(\la)= 0$. 
(We compute Kronecker coefficients with an adaption by 
J.~H\"uttenhain of a code originally written by H.~Derksen.) 
In this situation, there is (up to scaling) a unique 
highest weight function
$R\colon \Sym^3\C^9 \to \C$ of degree~$12$ and type $\la$. 
This function $R$ vanishes on $\Det_3$. 

Let us point out that the dimension of the 
``search space'' $\Sym^{12} \C^{165}$ in which $R$ lives is 
enormous: we have  
$\Sym^3\C^9\simeq \C^{165}$ and 
$\dim \Sym^{12} \C^{165} \approx 1.3 \cdot 10^{19}$. 
We have found the ``needle in a haystack'' with the help of 
represention theory and extensive calculations! 
It should also be emphasized that it is possible to describe $R$ in a concise way using 
symmetrizations, cf.~\cite{ike:12b}.
\end{example}

The following is a major open problem!  

\begin{problem}\label{problem:find-obstructions}
Find families of obstructions for forms with few rows.
\end{problem}

\subsection{Sketch of proof of Theorem~\ref{th:ms}}

\subsubsection{Symmetries of the determinant}

The symmetries of $\det_n$ are captured by the 
{\em stabilizer group} 
$$
 \stab_n := \Big\{ g\in \GL(\C^{n^2}) \mid 
    \det(g(X)) = \det(X) \Big\} ,
$$
where we interpret in this formula $X$ as a vector of length $n^2$. 
For $A,B\in\SL_n$ we consider the following linear map given by matrix multiplication: 
\begin{equation}\label{eq:gAB}
 g_{A,B} \colon \C^{n\times n} \to \C^{n\times n},\, 
 X \mapsto AXB .
\end{equation}
We have 
$\det(AXB) = \det(A)\det(X)\det(B) = \det(X)$. 
Hence $g_{A,B} \in \stab_n$.  
Are these all elements of the stabilizer group of $\det_n$? 
No, the transposition 
$\tau\colon\C^{n\times n} \to \C^{n\times n},\, X\mapsto X^T$
clearly also belongs to $\stab_n$.

The following result due to Frobenius~\cite{Frobdet} in fact 
states that each element of $\stab_n$ is of the form 
$g_{A,B}$ or $\tau g_{A,B}$. 
(This was rediscovered later by Dieudonn\'e~\cite{dieu:49}.) 
We skip the proof. 

\begin{theorem}[Frobenius]\label{th:frobenius} 
The stabilizer group $\stab_n$ of $\det_n$ is generated by 
$\tau$ and the $g_{A,B}$ for $A,B\in\SL_n$. 
We have 
$$
 \stab_n \simeq (\SL_n \times \SL_n) /\mu_n \rtimes \Z_2 ,
$$
where $\mu_n := \{ t\id_n \mid t^n =1 \}$. 
\end{theorem}

\subsubsection{Multiplicities in the coordinate ring of the orbit of $\det_n$}

In algebraic geometry, one defines a {\em regular function} 
$\varphi\colon\GL_{n^2}\cdot \det_n \to \C$ as a function 
such that each point of the orbit $\GL_{n^2}\cdot \det_n$ 
has an open neighborhood on which $\varphi$ can be 
expressed as the quotient of two rational functions. 
We denote by $\C[\GL_{n^2}\cdot \det_n]$ the ring of 
regular functions on the orbit. 

Let us point out that the orbit is a smooth algebraic variety, 
that is well understood in various senses. By going over 
to the orbit closure $\Det_n$, one adds limit points at 
the boundary, and we expect the situation to become very complicated. 
(Compare \cite{lamare:13,buik:15} for some results.) 

Clearly, we have the following  inclusion of rings of regular functions: 
$$
 \C[\Det_n] \subseteq \C\big[\GL_{n^2}\cdot \det_n\big] .
$$
By comparing multiplicities, it follows that for $\la\vdash_{n^2} dn$, 
\begin{eqnarray*}
\tilde{k}_n(\la) = \pleth_n(\la) - \mdet_n(\la) &=& \mbox{ multiplicity of $V_\la$ in $\C[\Det_n]$} \\
   & \le &  \mbox{ multiplicity of $V_\la$ in $\C[\GL_{n^2}\cdot \det_n]$} \\
   & = & \dim \big(V_\la \big)^{\stab_n} \quad\mbox{(algebraic Peter-Weyl theorem)}\\
   & \le & k_n(\la) \quad\mbox{(see below)}.
\end{eqnarray*}
The {\em Peter-Weyl theorem} is a well-known theorem from harmonic analysis,  
telling us about the irreducible $G$-modules in the space $L^2(G,\C)$ 
of quadratic integrable functions on a compact Lie group~$G$.  
(If $G$ is finite, this is just the well-known decomposition of the 
regular represention.) 
For the second equality above we used an
algebraic version of the Peter-Weyl theorem; 
cf. \cite[Chap.~II, sect.~3, Thm.~3]{kraf:84}
or \cite[Section~7.3]{procesi:07}.

We now justify the last inequality. It is here that Kronecker coefficients 
enter the game!
Schur-Weyl duality implies that by restricting the $\GL_{n^2}$-action of $V_\la(\GL_{n^2})$ 
with respect to the homomorphism 
$\GL_n\times\GL_n \to \GL_{n^2},\, (A,B) \mapsto A \ot B$, we obtain 
$$
 V_\la(\GL_{n^2}) \downarrow_{\GL_n\times\GL_n} = \bigoplus_{\mu,\nu\vdash_n |\la|}
  k(\la,\mu,\nu) \, V_\mu(\GL_n) \ot V_\nu(\GL_n) .
$$
We look now for $\SL_n\times\SL_n$-invariants. They occur on the 
right-hand side only if $\mu=\nu=n\times d$ and $|\la| = dn$. 
Note that $A\ot B$ is just another way of writing $g_{A,B}$; see \eqref{eq:gAB}.
Using Theorem~\ref{th:frobenius}, we obtain 
$$
 \dim \big(V_\la(\GL_{n^2})\big)^{\stab_n} \le \dim V_\la(\GL_{n^2})^{\SL_n\times\SL_n} 
   = k(\la,n\times d, n\times d) = k_n(\la) .
$$
This completes the proof of Theorem~\ref{th:ms}.

\subsection{Obstructions must be gaps} 

We address now the question of how to exhibit obstructions for forms. 
Example~\ref{ex:obstruction} was found with extensive calculations.
We will see here that, in a certain sense, obstructions are quite rare, 
or at least hard to find. 

Progress on Problem~\ref{problem:find-obstructions} is thus imperative. 
We don't want to hide the fact that we do not know whether there exist 
enough obstructions for achieving the desired lower bounds on determinantal complexity. 
In fact, the state of the art is that so far, no lower bound on $\dc(m)$ 
has been obtained along these lines. However, let us point out 
that in the related, but simpler situation of border rank  of tensors, 
lower bounds have been proven by exhibiting obstructions; see~\cite{BI:13}. 

We consider the following  set of highest weights
$$
 K_n := \big\{ \la \mid \mbox{$\la \vdash_{n^2} dn$ for some $d$ and $k_n(\la) > 0$} \big\} .
$$
From Definition~\ref{def:KronCoeff} it easily follows that 
$\la,\mu\in K_n$ implies $\la +\mu \in K_n$. 
Moreover, $0\in K_n$. Hence $K_n$ is a monoid.
(It follows from general principles that $K_n$ is 
finitely generated; cf.~\cite{brion:87}.)

\begin{example}
To illustrate the next step, consider the submonoid 
$M:=\{0,3,5,6,8,9,\ldots\}$ of $\N$, which clearly 
generates the group $\Z$. 
From $sx \in M$, $s\ge 1$, we cannot deduce that $x\in M$, 
due to the presence of the ``holes'' $1,2,4,7$. 
Filling in these holes, we obtain the monoid $\N$. 
The holes are usually called the {\em gaps of the monoid $M$}; cf.~\cite{alfonsin:05}. 
In general, one calls the process of filling in the gaps {\em saturation}. 
\end{example}

In our situation of interest, we make the following definition. 

\begin{definition}
The {\em saturation of $K_n$} is the set of partitions $\la$ 
with $\ell(\la)\le n^2$ such that 
$|\la|$ is a multiple of~$n$ and there exists a ``stretching factor'' 
$s\ge 1$ satisfying $s\la \in K_n$.
The {\em gaps} of $K_n$ are the elements in the saturation of $K_n$ 
that don't lie in $K_n$. 
\end{definition}

\begin{remark}
To fully justify the naming ``saturation'' here, one has to show that 
the group generated by $K_n$ consists of all $\la\in\Z^{n^2}$ such that 
$n$ divides $\sum_i \la_i$. (For $n\ge 7$ this was  shown in \cite{ik-pa:15}; 
for $n=2$ it is false.) 
\end{remark}

The following result is due to~\cite{buci:09}. 

\begin{theorem}[B, Christandl, Ikenmeyer]\label{th:buci}
The saturation of the monoid $K_n$ equals the 
set of all partitions $\la$ with $\ell(\la)\le n^2$ 
such that $|\la|$ is a multiple of $n$.
\end{theorem}

This result implies that obstructions 
must be gaps of the monoid $K_n$. 
The relevance of Theorem~\ref{th:buci} 
is that it excludes the use of asymptotic 
techniques for finding obstructions.
\medskip 

Theorem~\ref{th:ms} states that $\tilde{k}_n(\la) \le k_n(\la)$. 
However, we only need $\tilde{k}_n(\la) = 0$
for implementing our strategy of proving lower bounds. 
Indeed, 
the replacement of $\tilde{k}_n(\la)$ by the Kronecker coefficient~$k_n(\la)$ 
corresponds to replacing the coordinate ring of the orbit 
closure by the larger coordinate ring of the orbit, 
and this was only done because we better understand the latter.

So one might hope that Theorem~\ref{th:buci} 
fails for the smaller multiplicities $\tilde{k}_n$. 
Unfortunately, this doesn't turn out to be the case. 
Before stating the next result, we introduce a 
certain combinatorial conjecture. 
 
A {\em Latin square of size~$n$} is map $T\colon [n]^2 \to [n]$,
viewed as an $n\times n$ matrix with entries in~$[n]$, 
such that in each row and each column each entry in~$[n]$ 
appears exactly once.
So in each column and row we get a permutation of $[n]$. 
The column sign of $T$ is defined as the product of the signs of 
column permutations. The Latin square~$T$ is called 
{\em column-even} if this sign equals one, otherwise $T$ 
is called {\em column-odd}. 
See Figure~\ref{fig:LS} for an illustration.

\begin{figure}[h]
\begin{center}
$$\begin{array}{cccc}
-&+&-&-\\
1&2&3&4\\
4&1&2&3\\
3&4&1&2\\
2&3&4&1
\end{array}$$
\end{center}
\caption{\small A Latin square with column-sign $-1$.}\label{fig:LS}
\end{figure}

It is an easy exercise to check that if $n>1$ is odd, then there are 
as many column-even Latin squares of size~$n$ as there are 
column-odd Latin squares of size~$n$.

The Alon-Tarsi conjecture~\cite{AT:92} states that if $n$ is even, then  
the number of column-even Latin squares of size~$n$ is different from 
the number of column-odd Latin squares of size~$n$. 
This conjecture is known to be true if $n=p\pm 1$ where $p$ is a prime, 
cf.~\cite{Dri:98,Gly:10}.

The following result is due to Kumar~\cite{Kum:15}. 
(Note that, in contrast with Theorem~\ref{th:buci}, 
it only makes a statement about the $\la$ 
with $\ell(\la)\le n$.)

\begin{theorem}[Kumar]
If the Alon-Tarsi conjecture holds for $n$, then 
for all $\la$ with $\ell(\la) \le n$ such that $|\la|$ is a multiple of~$n$, 
we have $\tilde{k}_n(n\la) >0$.
\end{theorem}

In fact, it is possible to obtain an unconditional result 
at the price of losing the information about the specific 
stretching factor~$n$.
The following result is from~\cite{BHI:15}. 

\begin{theorem}[B, H\"uttenhain, Ikenmeyer]
For all $\la$ with $\ell(\la)\le n$ such that $|\la|$ is a multiple of~$n$, 
there exists $s\ge 1$ such that $\tilde{k}_n(s\la) >0$.
\end{theorem}

\section{Positivity of Kronecker coefficients}

Motivated by the attempt described in the previous section, 
notabable progress was made about understanding when Kronecker coefficients are positive.  
We report on this in the remainder of this survey. 

\subsection{Testing positivity is NP-hard}

It is known that testing the positivity of Littlewood\-Richardson coefficients 
can be done in polynomial time; cf~\cite{gct3,deloera:06,bu-ik:13}. 
Mulmuley conjectured~\cite{gct6} that testing positivity of Kronecker coefficients 
can be done in polynomial time as well. 
For {\em fixed~$m$} and partitions $\la,\mu,\nu$ of {\em length at most~$m$} 
this is true, see ~\cite{cdw:12}. 
However, an exciting recent result~\cite{ik-m-w:15}
shows that in general, this is not the case. 
For the following hardness results, we may even assume that the partitions are given as 
lists of integers encoded in unary.
(A positive integer $m$ encoded in unary has size $m$; thus considering unary 
encoding makes the problem easier.)

\begin{theorem}[Ikenmeyer, Mulmuley, Walter]\label{th:imw}
Testing positivity of  Kronecker coefficients is an NP-hard problem.

\end{theorem}


We are going to outline the proof.
By a {\em 3D-relation} we shall understand a finite subset $R$ 
of~$\N^3$. 
For $i\in\N$ we set 
$$
 x_R(i) := \#\{(x,y,z) \in R \mid x=i \} , 
$$ 
and we call the sequence $x_R:= (x_R(0),x_R(1),\ldots)$ 
the $x$-marginal of $R$. 
We may interpret $x_R$ as a partition of $|R]$ 
if the entries of $x_R$ are monotonically decreasing. 
(There is no harm caused by the fact that the indexing 
of $x_R$ starts with $0$.)
Similarly, we define the $y$-marginal $y_R$ and the 
$z$-marginal $z_R$ of $R$. 
Note that if $R$ is contained in the discrete cube $\{0,\ldots,m-1\}^3$, then 
$x_R,y_R,z_R$ have at most $m$ nonzero components.
The problem of reconstructing $R$ from its marginals 
is sometimes called ``discrete tomography''. 

We call a 3D-relation $R$ a {\em pyramid} if $(x,y,z) \in R$ implies 
$(x',y',z') \in R$ for all $(x',y',z') \in \N^3$ such that 
$x'\le x$, $y'\le y$, $z'\le z$. 
In the literature, one often calls pyramids {\em plane partitions}. 
In fact, they are just the 3D-analogues of Young diagrams.

Let $\la'$ denote the partition conjugate to $\la$ 
obtained by a reflection of its Young diagram at the main diagonal. 

\begin{definition}
For $\la,\mu,\nu\vdash d$ we denote by 
$t(\la,\mu,\nu)$ the number of 3D-relations $R$ with 
$x$-marginal $\la'$, $y$-marginal $\mu'$, and $z$-marginal $\nu'$.
Moreover, let $p(\la,\mu,\nu)$ denote the number of 
pyramids $R$ with the marginals $\la',\mu',\nu'$.
\end{definition}

The following result was previously proved by Manivel~\cite{mani:97} 
and rediscovered in~\cite{BI:13,ik-m-w:15}; compare also Vallejo~\cite{vallejo:00}.
 
\begin{lemma}\label{le:pkt}
We have $p(\la,\mu,\nu) \le  k(\la,\mu,\nu) \le t(\la,\mu,\nu)$
for $\la,\mu,\nu\vdash d$. 
\end{lemma}

\begin{proof}
Recall that $[\la'] \simeq [\la] \ot [1^d]$, where $d=|\la|$. 
Suppose that $\la',\mu',\nu'$ have at most $m$ parts. 
Then we have 
\begin{eqnarray*}
k(\la,\mu,\nu)  &=& \mult([\la] \ot [\mu] \ot [\nu], [d]) \\
 & = & \mult([\la'] \ot [\mu'] \ot [\nu'], [1^d]) \\
 & = & \mult\big(V_{\la'}(\GL_m) \ot V_{\mu'}(\GL_m) \ot V_{\nu'}(\GL_m), \Lambda^d(\C^m\ot\C^m\ot \C^m)\big) ,
\end{eqnarray*}
where for the last equality we have used Schur-Weyl duality. 

Let $e_j$ denote the $j$th canonical basis vector of $\C^m$. 
To a 3D-relation 
$R =\{ (x_i,y_i,z_i) \mid 1\le i \le d\} \subseteq \{0,\ldots,m-1\}^3$ 
such that $|R|=d$, we assign the vector (only defined up to sign) 
$$
  v_R := \pm\wedge_{i=1}^d (e_{x_i} \ot e_{y_i} \ot e_{z_i}) \ \in\ 
 \wedge^d (\C^m\ot\C^m\ot \C^m) .
$$
Note that the $v_R$ form a basis of $\wedge^d (\C^m\ot\C^m\ot \C^m)$. 
In fact, $v_R$ is a weight vector of weight $(x_R,y_R,z_R)$, since 
for a triple $g=(\diag(a_0,\ldots,a_{m-1}),\diag(b_0,\ldots,b_{m-1}), \diag(c_0,\ldots,c_{m-1}))$ 
of invertible diagonal matrices we have  
$$
 g\cdot v_R = 
 a_0^{x_R(0)}\cdots a_{m-1}^{x_R(m-1)}  b_0^{y_R(0)}\cdots b_{m-1}^{y_R(m-1)}  c_0^{z_R(0)}\cdots c_{m-1}^{z_R(m-1)}\,  v_R .
$$
We conclude that $t(\la,\mu,\nu)$ equals the dimension of the weight space of 
weight $(\la',\mu',\nu')$ in $\Lambda^d(\C^m\ot\C^m\ot \C^m)$. 

At the beginning of the proof, we observed that 
$k(\la,\mu,\nu)$ equals the multiplicity of 
$V_{\la'}(\GL_m) \ot V_{\mu'}(\GL_m) \ot V_{\nu'}(\GL_m)$
in $\Lambda^d(\C^m\ot\C^m\ot \C^m)$,
which is the dimension of the vector space of highest weight vectors
of weight $(\la',\mu',\nu')$ in $\Lambda^d(\C^m\ot\C^m\ot \C^m)$. 
So we conclude that $k(\la,\mu,\nu)\le t(\la,\mu,\nu)$. 

Finally, if $R$ is a pyramid, then it is easy to check that 
$(g_1,g_2,g_3) \cdot v_R = v_R$, where $g_1,g_2,g_3$ are invertible upper triangular matrices 
with $1$'s on the diagonal. In this case, $v_R$ is therefore a highest weight vector.  
This implies $p(\la,\mu,\nu)\le k(\la,\mu,\nu)$. 
\end{proof}

We will show now that certain constraints on the marginals of 
 a 3D-relation~$R$ enforce that $R$ must be a pyramid. 

The distance of the {\em barycenter} 
$b_R := \frac{1}{|R|}\sum_{p\in R} p$
of~$R$ to the linear hyperplane 
orthogonal to the diagonal $(1,1,1)$ 
is given by $h_R := b_R \cdot (1,1,1)^T$,  
up to the scaling factor $\sqrt{3}$. 
The distance~$h_R$ can be expressed in terms of the 
marginals of $R$ by  
\begin{equation}\label{eq:margy}
  \mbox{$ |R|\, h_R = \sum_{(x,y,z)\in R} (x+y+z) 
     = \sum_{i} i\, (x_R(i)  + y_R(i) + z_R(i) ) $.} 
\end{equation}
For $s\ge 1$ we consider the simplex 
$P(s) := \{ (x,y,z) \in\N^3 \mid x+y+z \le s-1\}$,  
which has the cardinality $|P(s)| = s(s+1)(s+2)/6$. 
For $d\ge 1$ we define $s(d)$ as the maximal natural number~$s$ such that 
$|P(s)| \le d$. 

Assume now that a 3D-relation $R$ satisfies 
$P(s)\subseteq R \subset P(s+1)$ for some $s$. 
Then necessarily $s=s(d)$, where $d := |R|$. 
In this situation, it is easy to see that $h_R = h(d)$, where 
\begin{equation}\label{eq:bary}
 \mbox{$h(d) :=  \frac{|P(s)|}{d}\,  h_{P(s)} + (1 - \frac{|P(s)|}{d} )\, s $.}
\end{equation}
If $\la',\mu',\nu'$ denote the marginals of $R$, then we have by \eqref{eq:margy}, 
\begin{equation}\label{def:simplex-like}
  \sum_{i}  i\, (\la'_i + \mu'_i +\nu'_i) = d\, h(d) .
\end{equation}
We call a triple $\la,\mu,\nu \vdash d$ of partitions  
{\em simplex-like} if \eqref{def:simplex-like} holds.

\begin{lemma}\label{le:p=k=t}
Any 3D-relation~$R$, whose marginals are simplex-like, is a pyramid.
Hence $k(\la,\mu,\nu) = t(\la,\mu,\nu)$ if $(\la,\mu,\nu)$ is simplex-like.
\end{lemma}

\begin{proof}
The first assertion is easy to prove  
and the second one follows with Lemma~\ref{le:pkt}.
\end{proof}

The following result was shown in~\cite{BDLG:00}. 

\begin{theorem}[Brunetti, Del Lungo, G\'erard]\label{th:brunetti}
Deciding $t(\la,\mu,\nu)>0$ is an NP-hard problem. 
\end{theorem}

The catch is that the reduction in the proof of this theorem 
from 3D-matching is such 
that one can actually reduce to simplex-like triples 
$(\la,\mu,\nu)$ of partitions. This completes 
our sketch of the proof of Theorem~\ref{th:imw}.
In fact, the NP-hardness reduction in the proof of Theorem~\ref{th:brunetti}
leads to an efficient and explicit way to produce many gaps of 
the Kronecker monoid. We are not aware of any other way to obtain this result! 
Unfortunately, the reduction breaks down for the most wanted situation of 
partition triples $(\la,\mu,\mu)$ where $\mu$ is a rectangle.
In fact, one can prove that $t(\la,n\times d, n\times d)>0$ if $\la\vdash dn$ 
such that $\ell(\la) \le\min\{d^2,n^2\}$, see \cite[Thm.~6.9]{ik-m-w:15}.
\medskip

From the proof of Theorem~\ref{th:brunetti} one 
obtains the following insights, which show a remarkable  
interplay between computer science and algebraic combinatorics.

\begin{itemize}
\item There is a positive $\sP$-formula for a subclass of 
triples of partitions, whose positivity of Kronecker coefficients
is NP-hard to decide.

\item The Kronecker monoid has many gaps,  
and we can efficiently compute subexponentially many of them. More specifically,
for any $0<\e <1$ there is $0<a<1$ such that for all~$m$, 
there exist $\Omega(2^{m^a})$ many partition triples 
$(\la,\mu,\mu)$ such that 
$k(\la,\mu,\mu) = 0$, but there exists $s\ge 1$ with 
$k(s\la,s\mu,s\mu) > 0$. Moreover, 
$\ell(\mu) \le m^\e$ and $|\la| = |\mu| \le m^3$. 
Finally, there is an efficient algorithm to produce these partitions. 
\end{itemize}

Since the reduction breaks down for the most wanted situation of 
partition triples $(\la,\mu,\mu)$ where $\mu$ is a rectangle, 
this fails to provide a solution for Problem~\ref{problem:find-obstructions}.

\subsection{Testing asymptotic positivity may be feasible} 

We finish by mentioning a further recent insight.

\begin{definition}
The {\em asymptotic positivity problem for Kronecker coefficients} 
is the problem of deciding for given $\la,\mu,\nu$ 
(in binary encoding) whether $k(s\la,s\mu,s\nu) > 0$ for some $s\ge 1$. 
\end{definition}

This problem can be rephrased as a membership problem to a (family of) 
polyhedral cones, that we may call {\em Kronecker cones}. 
They are of relevance for the quantum marginal problem 
of quantum information theory; see
\cite{klya:04,ChrMit06,ChristHarrowMitch07}. 

Theorem~\ref{th:imw} states that the positivity testing problem 
for Kronecker coefficients is NP-hard. 
By contrast, the following recent result~\cite{bcmw:15}
tells us that the asymptotic version of this problem 
should be considerably easier.

\begin{theorem}[B, Christandl, Mulmuley, Walter]\label{th:aspos}
The asymptotic positivity problem for Kronecker coefficients is in 
$\NP\cap\coNP$. 
\end{theorem}

In fact, we have now good reasons to conjecture that the asymptotic positivity problem 
for Kronecker coefficients can be solved in polynomial time. 
In view of the known algorithms and the complicated face structure 
of the Kronecker cones~\cite{ressayre:10,verwa:14}, this is quite surprising.

The proof of Theorem~\ref{th:aspos} 
combines different techniques. 
The containment in $\NP$ is a consequence of the description of the Kronecker cone as the
image of the so-called moment map, which is a consequence of a general result due 
to Mumford~\cite{ness:84}; see also \cite{brion:87}. 
Moment maps are studied in symplectic geometry. 

The basis of the containment in $\coNP$ is a description of the facets 
of the Kronecker cone due to Ressayre~\cite{ressayre:10}. 
Vergne and Walter~\cite{verwa:14} provided a modification of Ressayre's 
description that is efficiently testable, which leads to the containement in $\coNP$.

\section{Note added in proof}

Since the writing of this survey in the fall of 2015, important progress has been made 
with regard to the feasibility of the attempt outlined in Section~\ref{se:attempt}. 

In a breakthrough work,
Ikenmeyer and Panova~\cite{ik-pa:15} showed that 
the vanishing of rectangular Kronecker coefficients 
cannot be used to prove superpolynomial lower bounds on the 
determinantal complexity of the permanent polynomials! 

Recall that by Theorem~\ref{th:kl}, an occurence obstruction~$\la$ 
proving $\dc(m) > n$ necessarily satisfies 
$\ell(\la) \le m^2 +1$ and $\la_1 \ge |\la| ( 1- m/n)$. 
(By a minor modification of the notion of padded permanents, 
we may even assume $\ell(\la) \le m^2$.)  


More specifically, Ikenmeyer and Panova proved the following.

\begin{theorem}[Ikenmeyer and Panova]\label{th:IP}
Let $\la\vdash dn$ such that $\ell(\la) \le m^2$,
$\la_1 \ge |\la| ( 1- m/n)$, and assume $n>3m^4$. 
Then $\pleth_n(\la) >0$ implies $k_n(\la) >0$.
\end{theorem}

This result does not yet rule out the occurence based approach 
towards $\VP\ne \VNP$ as outlined in Section~\ref{se:attempt}, 
since it refers to the Kronecker coefficients $k_n(\la)$ of 
rectangular partitions and not to the
GCT-coefficients $\tilde{k}_n(\la)$.
(Recall those are the multiplicities 
in the coordinate ring of the orbit closure of $\Det_n$; 
see \eqref{def:ktilde} and Theorem~\ref{th:ms}.) 

However, shortly after the appearance of~\cite{ik-pa:15}, 
B\"urgisser, Ikenmeyer and Panova~\cite{bu-ik-pa:16} proved 
a similarly devastating result for the GCT-coefficients.

\begin{theorem}[B, Ikenmeyer, and Panova]\label{th:BIP}
Let $\la\vdash dn$ such that we have $\ell(\la) \le m^2$, 
$\la_1 \ge |\la| ( 1- m/n)$, and assume $n>m^{25}$. 
Then $\pleth_n(\la) >0$ implies $\tilde{k}_n(\la) >0$.
\end{theorem}


The main ingredient behind the proof of Theorem~\ref{th:BIP}, 
besides a splitting technique as for Theorem~\ref{th:IP}, 
is the encoding of a generating system of highest weight vectors 
in plethysms $\Sym^d\Sym^n V$ by (classes of) tableaux with contents $d\times n$, as well as 
the analysis of their evaluation at tensors of rank one in a combinatorial way. 
This is similar to \cite{bci:10,ike:12b}.  
A further technique is the ``lifting'' of highest weight vectors of $\Sym^d\Sym^n V$, 
when increasing the inner degree~$n$, 
as introduced by Kadish and Landsberg~\cite{kadish-landsberg:14}. 
This is closely related to stability property of the plethysm coefficients
\cite{wei:90,carre-thibon:92,mani:98}. 
Remarkably, for the proof of Theorem~\ref{th:BIP}, the only information 
needed about the orbit closures $\Det_n$ is that they contain certain padded power sums, 
see also~\cite{buik:15}.

\subsection{Final remarks} 

Unfortunately, Theorem~\ref{th:BIP} rules out the possibility of proving 
$\VP\ne \VNP$ via occurrence obstructions.

Let us emphasize that there still remains the possibility that 
the approach via representation theoretic obstructions may succeed 
when comparing multiplicities. 
Indeed, if the orbit closure $Z_{n,m}$ of the padded permanent
$t^{n-m} \per_m$ is contained in $\Det_n$, then 
the restriction defines a surjective $\GL_{n^2}$-equivariant homomorphism 
$\C[\Omega_n]\to\C[Z_{n,m}]$ of the coordinate rings, 
and hence  the multiplicity of the type $\la$ in 
$\C[Z_{n,m}]$ is bounded from above by
the GCT-coefficient $\tilde{k}_n(\la)$. 
Thus, proving that $\tilde{k}_n(\la)$ is strictly smaller 
than the latter multiplicity implies that  
$Z_{n,m}\not\subseteq\Det_n$. 
Mumuley pointed out to us a paper by 
Larsen and Pink~\cite{lapa:90} that is of 
potential interest in this connection. 

In this context let us remark that~\cite{cdw:12} shows  
that comparing multiplicities by asymptotic methods 
cannot be sufficient for the purpose of complexity separation. 

To conclude, even if the approach via multiplicities should turn out to be impossible 
as well, we should keep in mind that the noncontainment of orbit closures 
in principle can be proved by exhibiting some highest weight vector functions
(see \cite[Prop.~3.3]{BI:13}). Classic invariant theory and representation theory 
should provide guidelines on how to find such functions, even though 
our current understanding of this is very limited.

\bigskip

\noindent{\bf Acknowledgements.} 
I thank Jesko H\"uttenhain, Christian Ikenmeyer, Joseph Landsberg, and Ketan Mulmuley 
for their feedback. Special thanks go to Christian Krattenthaler for his detailed comments 
on the manuscript. I am grateful to the Simons Institute for the Theory of Computing 
in Berkeley for making possible the semester program 
``Algorithms and Complexity in Algebraic Geometry'', 
that provided ideal conditions for achieving the recent progress outlined in this survey.


\def\cprime{$'$}


\begin{thebibliography}{10}

\bibitem{AT:92}
Noga Alon and Michael Tarsi.
\newblock Colorings and orientations of graphs.
\newblock {\em Combinatorica}, 12(2):125--134, 1992.

\bibitem{al-bo-ve:15}
Jarod Alper, Tristram Bogart, and Mauricio Velasco.
\newblock A lower bound for the determinantal complexity of hypersurface.
\newblock {\em Found. Comput. Math.}, 2016.
\newblock To appear.

\bibitem{blasiak:12}
Jonah Blasiak.
\newblock Kronecker coefficients for one hook shape.
\newblock arXiv:1209.2018, 2012.

\bibitem{brion:87}
Michel Brion.
\newblock Sur l'image de l'application moment.
\newblock In {\em S\'eminaire d'alg\`ebre {P}aul {D}ubreil et {M}arie-{P}aule
  {M}alliavin ({P}aris, 1986)}, volume 1296 of {\em Lecture Notes in Math.},
  pages 177--192. Springer, Berlin, 1987.

\bibitem{BDLG:00}
Sara Brunetti, Alberto Del~Lungo, and Yan G\'{e}rard.
\newblock On the computational complexity of reconstructing three-dimensional
  lattice sets from their two-dimensional {X}-rays.
\newblock In {\em Proceedings of the {W}orkshop on {D}iscrete {T}omography:
  {A}lgorithms and {A}pplications ({C}ertosa di {P}ontignano, 2000)}, volume
  339, pages 59--73, 2001.

\bibitem{buer:00-3}
Peter B{\"u}rgisser.
\newblock {\em Completeness and Reduction in Algebraic Complexity Theory},
  volume~7 of {\em {A}lgorithms and {C}omputation in {M}athematics}.
\newblock Springer Verlag, 2000.

\bibitem{bci:10}
Peter B\"urgisser, Matthias Christandl, and Christian Ikenmeyer.
\newblock Even partitions in plethysms.
\newblock {\em Journal of Algebra}, 328(1):322 -- 329, 2011.

\bibitem{buci:09}
Peter B{\"u}rgisser, Matthias Christandl, and Christian Ikenmeyer.
\newblock Nonvanishing of {K}ronecker coefficients for rectangular shapes.
\newblock {\em Adv. Math.}, 227(5):2082--2091, 2011.

\bibitem{bcmw:15}
Peter B\"urgisser, Matthias Christandl, Ketan Mulmuley, and Michael Walter.
\newblock On the complexity of the membership problem for moment polytopes.
\newblock arXiv:1511.03675, 2015.

\bibitem{BHI:15}
Peter B\"urgisser, Jesko H\"uttenhain, and Christian Ikenmeyer.
\newblock Permanent versus determinant: not via saturations.
\newblock arXiv:1501.05528, 2015.

\bibitem{buik:15}
Peter B\"urgisser and Christian Ikenmeyer.
\newblock Fundamental invariants of orbit closures.
\newblock arXiv:1511.02927, 2015.

\bibitem{bu-ik:13}
Peter B{\"u}rgisser and Christian Ikenmeyer.
\newblock Deciding positivity of {L}ittlewood-{R}ichardson coefficients.
\newblock {\em SIAM J. Discrete Math.}, 27(4):1639--1681, 2013.

\bibitem{BI:13}
Peter B{\"u}rgisser and Christian Ikenmeyer.
\newblock Explicit lower bounds via geometric complexity theory.
\newblock {\em Proceedings 45th Annual ACM Symposium on Theory of Computing
  2013}, pages 141--150, 2013.

\bibitem{bu-ik-pa:16}
Peter B\"{u}rgisser, Christian Ikenmeyer, and Greta Panova.
\newblock No occurrence obstructions in geometric complexity theory.
\newblock arXiv:1604.06431, 2016.

\bibitem{BLMW:11}
Peter B\"urgisser, J.M. Landsberg, Laurent Manivel, and Jerzy Weyman.
\newblock An overview of mathematical issues arising in the {G}eometric
  complexity theory approach to {VP} v.s. {VNP}.
\newblock {\em {SIAM} J. Comput.}, 40(4):1179--1209, 2011.

\bibitem{cai:2010}
Jin-Yi Cai, Xi~Chen, and Dong Li.
\newblock Quadratic lower bound for permanent vs. determinant in any
  characteristic.
\newblock {\em Comput. Complexity}, 19(1):37--56, 2010.

\bibitem{carre-thibon:92}
Christophe Carr{\'e} and Jean-Yves Thibon.
\newblock Plethysm and vertex operators.
\newblock {\em Adv. in Appl. Math.}, 13(4):390--403, 1992.

\bibitem{cdw:12}
Matthias Christandl, Brent Doran, and Michael Walter.
\newblock Computing multiplicities of {L}ie group representations.
\newblock In {\em 2012 {IEEE} 53rd {A}nnual {S}ymposium on {F}oundations of
  {C}omputer {S}cience---{FOCS} 2012}, pages 639--648. IEEE Computer Soc., Los
  Alamitos, CA, 2012.

\bibitem{ChrMit06}
Matthias Christandl and Graeme Mitchison.
\newblock The spectra of density operators and the {K}ronecker coefficients of
  the symmetric group.
\newblock {\em Comm. Math. Phys.}, 261(3):789--797, February 2006.

\bibitem{ChristHarrowMitch07}
Matthias Christiandl, Aram Harrow, and Graeme Mitchison.
\newblock On nonzero {K}ronecker coefficients and what they tell us about
  spectra.
\newblock {\em Comm. Math. Phys.}, 270(3):575--585, 2007.

\bibitem{dieu:49}
Jean Dieudonn{\'e}.
\newblock Sur une g\'en\'eralisation du groupe orthogonal \`a quatre variables.
\newblock {\em Arch. Math.}, 1:282--287, 1949.

\bibitem{draisma-siam:15}
Jan Draisma.
\newblock Geometry, invariants, and the elusive search for complexity lower
  bounds.
\newblock {\em SIAM News}, 48(6), 2015.

\bibitem{Dri:98}
Arthur~A. Drisko.
\newblock Proof of the {A}lon-{T}arsi conjecture for {$n=2^rp$}.
\newblock {\em Electron. J. Combin.}, 5:Research paper 28, 5 pp. (electronic),
  1998.

\bibitem{Frobdet}
Georg Frobenius.
\newblock {\"U}ber die {D}arstellung der endlichen {G}ruppen durch lineare
  {S}ubstitutionen.
\newblock {\em Sitzungsber Deutsch. Akad. Wiss. Berlin}, pages 994--1015, 1897.

\bibitem{FH:91}
William Fulton and Joe Harris.
\newblock {\em Representation Theory - A First Course}, volume 129 of {\em
  Graduate Texts in Mathematics}.
\newblock Springer, 1991.

\bibitem{Gly:10}
David~G. Glynn.
\newblock The conjectures of {A}lon-{T}arsi and {R}ota in dimension prime minus
  one.
\newblock {\em SIAM J. Discrete Math.}, 24(2):394--399, 2010.

\bibitem{grenet:11}
Bruno Grenet.
\newblock An upper bound for the permanent versus determinant problem.
\newblock Accepted for {\it Theory of Computing}, 2011.

\bibitem{lahu:16}
Jesko H{\"u}ttenhain and Pierre Lairez.
\newblock The boundary of the orbit of the 3 by 3 determinant polynomial.
\newblock arXiv:1512.02437, 2015.

\bibitem{ike:12b}
Christian Ikenmeyer.
\newblock {\em Geometric Complexity Theory, Tensor Rank, and
  {L}ittlewood-{R}ichardson Coefficients}.
\newblock PhD thesis, Institute of Mathematics, University of Paderborn, 2012.

\bibitem{hu-ik:14}
Christian Ikenmeyer and Jesko H\"{u}ttenhain.
\newblock Binary determinantal complexity.
\newblock {\em Linear Algebra and its Applications}, 504:559--573, 2016.

\bibitem{ik-m-w:15}
Christian Ikenmeyer, Ketan Mulmuley, and Michael Walter.
\newblock On vanishing of {K}ronecker coefficients.
\newblock arXiv:1507.02955, 2015.

\bibitem{ik-pa:15}
Christian Ikenmeyer and Greta Panova.
\newblock Rectangular {K}ronecker coefficients and plethysms in geometric
  complexity theory.
\newblock arXiv:1512.03798, 2015.

\bibitem{kadish-landsberg:14}
Harlan Kadish and J.~M. Landsberg.
\newblock Padded polynomials, their cousins, and geometric complexity theory.
\newblock {\em Comm. Algebra}, 42(5):2171--2180, 2014.

\bibitem{klya:04}
Alexander Klyachko.
\newblock Quantum marginal problem and representations of the symmetric group.
\newblock arXiv:quant-ph/0409113, 2003.

\bibitem{kraf:84}
Hanspeter Kraft.
\newblock {\em Geometrische {M}ethoden in der {I}nvariantentheorie}.
\newblock Aspects of Mathematics, D1. Friedr. Vieweg \& Sohn, Braunschweig,
  1984.

\bibitem{Kum:15}
Shrawan Kumar.
\newblock A study of the representations supported by the orbit closure of the
  determinant.
\newblock {\em Compositio Mathematica}, 151:292--312, 2 2015.

\bibitem{lare:15}
Joseph Landsberg and Nicolas Ressayre.
\newblock Permanent v.\ determinant: an exponential lower bound assuming
  symmetry and a potential path towards {V}aliant's conjecture.
\newblock arXiv:1508.05788, 2015.

\bibitem{lamare:13}
Joseph~M. Landsberg, Laurent Manivel, and Nicolas Ressayre.
\newblock Hypersurfaces with degenerate duals and the geometric complexity
  theory program.
\newblock {\em Comment. Math. Helv.}, 88(2):469--484, 2013.

\bibitem{lapa:90}
M.~Larsen and R.~Pink.
\newblock Determining representations from invariant dimensions.
\newblock {\em Invent. math.}, 102:377--389, 1990.

\bibitem{deloera:06}
Jes{\'u}~De Loera and Tyrrell McAllister.
\newblock On the computation of {C}lebsch-{G}ordan coefficients and the
  dilation effect.
\newblock {\em Experiment. Math.}, 15(1):7--19, 2006.

\bibitem{mapo:08}
Guillaume Malod and Natacha Portier.
\newblock Characterizing {V}aliant's algebraic complexity classes.
\newblock {\em Journal of Complexity}, 24:16--38, 2008.

\bibitem{mani:98}
L.~Manivel.
\newblock Gaussian maps and plethysm.
\newblock In {\em Algebraic geometry ({C}atania, 1993/{B}arcelona, 1994)},
  volume 200 of {\em Lecture Notes in Pure and Appl. Math.}, pages 91--117.
  Dekker, New York, 1998.

\bibitem{mani:97}
Laurent Manivel.
\newblock Applications de {G}auss et pl\'ethysme.
\newblock {\em Ann. Inst. Fourier (Grenoble)}, 47(3):715--773, 1997.

\bibitem{marcus-minc:61}
Marvin Marcus and Henryk Minc.
\newblock On the relation between the determinant and the permanent.
\newblock {\em Illinois J. Math.}, 5:376--381, 1961.

\bibitem{mignon-ressayre:04}
Thierry Mignon and Nicolas Ressayre.
\newblock A quadratic bound for the determinant and permanent problem.
\newblock {\em Int. Math. Res. Not.}, 79:4241--4253, 2004.

\bibitem{gct6}
Ketan~D. Mulmuley.
\newblock Geometric complexity theory {VI}: the flip via positivity.
\newblock Technical report, Computer Science Department, the University of
  Chicago, 2010.

\bibitem{gct3}
Ketan~D. Mulmuley, Hariharan Narayanan, and Milind Sohoni.
\newblock Geometric complexity theory {III}: on deciding nonvanishing of a
  {L}ittlewood-{R}ichardson coefficient.
\newblock {\em J. Algebraic Combin.}, 36(1):103--110, 2012.

\bibitem{gct1}
Ketan~D. Mulmuley and Milind Sohoni.
\newblock Geometric complexity theory. {I}. {A}n approach to the {P} vs.\ {NP}
  and related problems.
\newblock {\em SIAM J. Comput.}, 31(2):496--526 (electronic), 2001.

\bibitem{gct2}
Ketan~D. Mulmuley and Milind Sohoni.
\newblock Geometric complexity theory. {II}. {T}owards explicit obstructions
  for embeddings among class varieties.
\newblock {\em SIAM J. Comput.}, 38(3):1175--1206, 2008.

\bibitem{mumf:95}
David Mumford.
\newblock {\em Algebraic geometry. {I}}.
\newblock Classics in Mathematics. Springer-Verlag, Berlin, 1995.
\newblock Complex projective varieties, Reprint of the 1976 edition.

\bibitem{ness:84}
Linda Ness.
\newblock A stratification of the null cone via the moment map.
\newblock {\em Amer. J. Math.}, 106(6):1281--1329, 1984.
\newblock With an appendix by David Mumford.

\bibitem{polya:13}
Georg P{\'o}lya.
\newblock Aufgabe 424.
\newblock {\em Arch. Math. Phys.}, 20:271, 1913.

\bibitem{procesi:07}
Claudio Procesi.
\newblock {\em Lie groups}.
\newblock Universitext. Springer, New York, 2007.
\newblock An approach through invariants and representations.

\bibitem{alfonsin:05}
Jorge~Luis Ram{\'{\i}}rez~Alfons{\'{\i}}n.
\newblock {\em The {D}iophantine {F}robenius problem}, volume~30 of {\em Oxford
  Lecture Series in Mathematics and its Applications}.
\newblock Oxford University Press, Oxford, 2005.

\bibitem{ressayre:10}
Nicolas Ressayre.
\newblock Geometric invariant theory and the generalized eigenvalue problem.
\newblock {\em Invent. Math}, 180:389--441, 2010.

\bibitem{schur:68}
Issai Schur.
\newblock {\em Vorlesungen \"uber {I}nvariantentheorie}.
\newblock Bearbeitet und herausgegeben von Helmut Grunsky. Die Grundlehren der
  mathematischen Wissenschaften, Band 143. Springer-Verlag, Berlin-New York,
  1968.

\bibitem{sturmfels:93}
Bernd Sturmfels.
\newblock {\em Algorithms in invariant theory}.
\newblock Texts and Monographs in Symbolic Computation. Springer-Verlag,
  Vienna, 1993.

\bibitem{szego:13}
G\'abor Szeg{\H o}.
\newblock Zu {A}ufgabe 424.
\newblock {\em Arch. Math. Phys.}, 21:291--292, 1913.

\bibitem{Val:79b}
Leslie~G. Valiant.
\newblock Completeness classes in algebra.
\newblock In {\em Conference {R}ecord of the {E}leventh {A}nnual {ACM}
  {S}ymposium on {T}heory of {C}omputing ({A}tlanta, {G}a., 1979)}, pages
  249--261. ACM, New York, 1979.

\bibitem{vallejo:00}
Ernesto Vallejo.
\newblock Plane partitions and characters of the symmetric group.
\newblock {\em J. Algebraic Combin.}, 11(1):79--88, 2000.

\bibitem{verwa:14}
Mich{\`e}le Vergne and Michael Walter.
\newblock Inequalities for moment cones of finite-dimensional representations.
\newblock arXiv:1410.8144, 2015.

\bibitem{wei:90}
Steven~H. Weintraub.
\newblock Some observations on plethysms.
\newblock {\em J. Algebra}, 129(1):103--114, 1990.

\bibitem{yabe:15}
Akihiro Yabe.
\newblock Bi-polynomial rank and determinantal complexity.
\newblock arXiv:1504.00151, 2015.

\end{thebibliography}

\end{document}